\documentclass[letter,11pt]{article}

\usepackage[bottom=1in,top=1in,left=1in,right=1in]{geometry}
\usepackage{url,hyperref}
\usepackage{tikz}

\usepackage{amsmath}
\usepackage{amsthm}
\usepackage{graphicx}
\usepackage{epsfig}
\usepackage{subfigure}
\usepackage{natbib}

\usepackage{amsfonts}
\usepackage{algorithm}
\usepackage[noend]{algpseudocode}
\DeclareMathOperator*{\argmax}{arg\,max}
\DeclareMathOperator*{\argmin}{arg\,min}

\theoremstyle{plain}
\newtheorem{theorem}	 			{Theorem}
\newtheorem{lemma}		[theorem]	{Lemma}

\newtheorem{corollary}		[theorem]	{Corollary} 
 
\theoremstyle{definition}
\newtheorem{definition}[theorem]{Definition} 
\theoremstyle{remark}

\newcommand{\val}{\mathrm{val}}
\newcommand{\tra}{^T}
\newcommand{\NN}{\mathbb{N}}
\newcommand{\RR}{\mathbb{R}}
\newcommand{\RRp}{\mathbb{R}_{\ge0}}
\newcommand{\calF}{\ensuremath{\mathcal{F}}}
\newcommand{\calI}{\ensuremath{\mathcal{I}}}
\newcommand{\calL}{\ensuremath{\mathcal{L}}}
\newcommand{\calP}{\ensuremath{\mathcal{P}}}
\newcommand{\calQ}{\ensuremath{\mathcal{Q}}}
\newcommand{\calS}{\ensuremath{\mathcal{S}}}
\newcommand{\eps}{\varepsilon}
\newcommand{\Ex}[2][]{\text{\rm\bf E}_{#1}\hspace{-0.03cm}\left[#2\right]}
\newcommand{\dist}{\mathrm{dist}}

\title{Smoothed Analysis of Pareto Curves\\ in Multiobjective Optimization}
\author{Heiko R{\"o}glin\thanks{Department of Computer Science, University of Bonn, Bonn, Germany, \url{roeglin@cs.uni-bonn.de}}}
\date{}

\begin{document}

\maketitle

\begin{abstract}
In a multiobjective optimization problem a solution is called Pareto-optimal if no criterion can be improved without deteriorating at least one of the other criteria. Computing the set of all Pareto-optimal solutions is a common task in multiobjective optimization to filter out unreasonable trade-offs.

For most problems the number of Pareto-optimal solutions increases only moderately with the input size in applications. However, for virtually every multiobjective optimization problem there exist worst-case instances with an exponential number of Pareto-optimal solutions. In order to explain this discrepancy, we analyze a large class of multiobjective optimization problems in the model of smoothed analysis and prove a polynomial bound on the expected number of Pareto-optimal solutions.

We also present algorithms for computing the set of Pareto-optimal solutions for different optimization problems and discuss related results on the smoothed complexity of optimization problems.
\end{abstract}

\section{Algorithms for Computing Pareto Curves}

Suppose you would like to book a flight to your favorite conference. Your decision is then probably guided by different factors, like the price, the number of stops, and the arrival time. Usually you won't find a flight that is optimal in every respect and you have to choose the best trade-off. This is characteristic for many decisions faced every day by people, companies, and other economic entities.

The notion of ``best trade-off'' is hard to formalize and often there is no consensus on how different criteria should be traded off against each other. However, there is little disagreement that in a reasonable outcome no criterion can be improved without deteriorating at least one of the other criteria. Outcomes with this property are called \emph{Pareto-optimal} and they play a crucial role in multi-criteria decision making as they help to filter out unreasonable solutions. In this section we discuss algorithms for computing the set of Pareto-optimal solutions for different problems.

\subsection{Knapsack Problem}
The \emph{knapsack problem} is a well-known NP-hard
optimization problem. An instance of this problem consists of
a set of items, each with a profit and a weight, and a capacity.
The goal is to find a subset of the items that maximizes the total profit among all
subsets whose total weight does not exceed the capacity.
Let~$p=(p_1,\ldots,p_n)\tra\in\RRp^n$ and~$w=(w_1,\ldots,w_n)\tra\in\RRp^n$ denote the profits and weights,
respectively, and let~$W\in\RRp$ denote the capacity.
Formally the knapsack problem can be stated as follows:
\begin{align*}
   \text{maximize \hspace{0.3cm}} & p\tra x = p_1x_1+\cdots+p_nx_n\\
   \text{subject to \hspace{0.3cm}} & w\tra x = w_1x_1+\cdots+w_nx_n\le W,\\
   & \text{and~} x=(x_1,\ldots,x_n)\tra\in\{0,1\}^n.
\end{align*}

The knapsack problem has attracted a great deal of
attention, both in theory and in practice. Theoreticians are interested
in the knapsack problem because of its simple structure; it can be
expressed as a binary program with one linear objective function and
one linear constraint. On the other hand, knapsack-like problems
often occur in applications, and practitioners have
developed numerous heuristics for solving them. These heuristics work
very well on random and real-world instances 
and they usually find optimal solutions quickly even for very large instances.

In the following, we assume that an arbitrary instance~$\calI$ of the knapsack problem is given.
We use the term \emph{solution} to refer to a vector
$x\in\{0,1\}^n$, and we say that a solution is \emph{feasible} if~$w\tra x\le W$.
We say that a solution $x$ \emph{contains item~$i$} if $x_i=1$
and that it \emph{does not contain item~$i$} if~$x_i=0$.

One naive approach for solving the knapsack problem is to enumerate
all feasible solutions and to select the one with maximum
profit. This approach is not efficient as there are
typically exponentially many feasible solutions. In order to decrease
the number of solutions that have to be considered, we view the knapsack problem as a bicriteria optimization problem and restrict the enumeration to only the Pareto-optimal solutions. 

\begin{definition}\label{def:ParetoOptimal}
A solution~$y$ \emph{dominates} a solution~$x$ if~$p\tra y\ge p\tra x$
and~$w\tra y\le w\tra x$, with at least one of these inequalities being strict.
A solution $x$ is called \emph{Pareto-optimal} if it is not dominated by any other solution.
The \emph{Pareto set} or \emph{Pareto curve} is the set of all Pareto-optimal solutions.
\end{definition}

Once the Pareto set is known, the given instance of the knapsack problem can be solved optimally in time linear in the size of this set due to the following observation.
\begin{lemma}
There always exists an optimal solution that is also Pareto-optimal.
\end{lemma}

\begin{proof}

Take an arbitrary optimal solution~$x$ and assume that it is not Pareto-optimal. There cannot be a solution~$y$ with $p\tra y > p\tra x$ and $w\tra y \le w\tra x$ because then~$y$ would be a better solution than~$x$. Hence, if $x$ is not Pareto-optimal then it is dominated by a solution~$y$ with $p\tra y = p\tra x$ and $w\tra y < w\tra x$. Then either~$y$ is Pareto-optimal or we repeat the argument to find a solution~$z$ with $p\tra z = p\tra y$ and $w\tra z < w\tra y$. This construction terminates after a finite number of iterations with an optimal solution that is also Pareto-optimal.
\end{proof}

We denote the Pareto set by~$\calP\subseteq\{0,1\}^n$. It can happen that there are two or more Pareto-optimal solutions with the same profit and the same weight. Then~$\calP$ is assumed to contain only one of these solutions, which can be chosen arbitrarily. Due to the previous lemma the solution
\[
   x^{\star} = \argmax_{x\in\calP} \{p\tra x \mid w\tra x \le W\},
\]
is an optimal solution of the given instance of the knapsack problem.

In the following we present an algorithm invented by \citet{NU69} to compute the Pareto set of a given instance of the knapsack problem. We will refer to this algorithm, which is based on dynamic programming, as the \emph{Nemhauser-Ullmann algorithm}. For each~$i\in\{0,1,\ldots,n\}$ it computes the Pareto set~$\calP_i$ of the restricted instance~$\calI_i$ that contains only the first~$i$ items of the given instance~$\calI$. Then~$\calP_n=\calP$ is the set we are looking for.
Let
\[
   \calS_i=\{x\in\{0,1\}^n\mid x_{i+1}=\ldots=x_n = 0\}
\]
denote the set of solutions that do not contain the items~$i+1,\ldots,n$.
Formally, solutions of the instance~$\calI_i$ are binary vectors of length~$i$. We will, however, represent them as binary vectors of length~$n$ from~$\calS_i$.
For a solution~$x\in\{0,1\}^n$ and an item~$i\in\{1,\ldots,n\}$ we denote
by~$x^{+i}$ the solution that is obtained by adding item~$i$ to solution~$x$:
\[
   x^{+i}_j = \begin{cases}
   	  x_j & \text{if $j\neq i$,}\\
   	  1   & \text{if $j=i$.}
   \end{cases}
\]
Furthermore, for a set~$\calS\subseteq\{0,1\}^n$ of solutions let
\[
   \calS^{+i} = \{y\in\{0,1\}^n \mid \exists x\in\calS: y = x^{+i}\}.
\]
If for some~$i\in\{1,\ldots,n\}$, the set~$\calP_{i-1}$ is known then the set~$\calP_i$ can be computed with the help of the following lemma.
For the lemma we assume a consistent tie-breaking between solutions that have the same profit and the same weight. In particular, if $p\tra x = p\tra y$ and $w\tra x = w\tra y$ for two solutions~$x$ and~$y$ and the tie-breaking favors~$x$ over~$y$ then it should also favor~$x^{+i}$ over $y^{+i}$ for any~$i$.
\begin{lemma}\label{lem:calPi}
For every~$i\in\{1,\ldots,n\}$, the set $\calP_i$ is a subset of $\calP_{i-1} \cup \calP_{i-1}^{+i}$.
\end{lemma}
\begin{proof}
Let~$x\in\calP_i$. Based on the value of~$x_i$ we distinguish two cases.

First we consider the case~$x_i=0$. We claim that in this case~$x\in\calP_{i-1}$.
Assume for contradiction that~$x\notin\calP_{i-1}$. Then there exists a
solution~$y\in\calP_{i-1}\subseteq\calS_{i-1}\subseteq\calS_i$ that dominates~$x$. Since~$y\in\calS_i$, solution~$x$ cannot be
Pareto-optimal among the solutions in~$\calS_i$. Hence, $x\notin\calP_i$, contradicting
the choice of~$x$.

Now we consider the case~$x_i=1$. We claim that in this case~$x\in\calP_{i-1}^{+i}$.
Since~$x\in\calS_i$ and~$x_i=1$, there exists a solution~$y\in\calS_{i-1}$ such that~$x=y^{+i}$.
We need to show that~$y\in\calP_{i-1}$. Assume for contradiction that there exists a solution~$z\in\calP_{i-1}$
that dominates~$y$. Then~$p\tra z \ge p\tra y$ and~$w\tra z \le w\tra y$ and one of these
inequalities is strict. By adding item~$i$ to the solutions~$y$ and~$z$, we obtain~$p\tra z^{+i} \ge p\tra y^{+i}$ and~$w\tra z^{+i} \le w\tra y^{+i}$, with one of these
inequalities being strict. Hence, the solution~$z^{+i}$ dominates the solution~$x=y^{+i}$. Since~$z^{+i}\in\calS_i$,
this implies~$x\notin\calP_i$, contradicting the choice of~$x$.
\end{proof}

Due to the previous lemma, the Pareto set~$\calP_i$ can be computed easily if the Pareto set~$\calP_{i-1}$ is already known.
For this one only needs to compute the set~$\calP_{i-1} \cup \calP_{i-1}^{+i}$ and remove solutions from this set that
are dominated by other solutions from this set. Using additionally that~$\calP_0=\calS_0=\{0^n\}$, we obtain the following
algorithm to solve the knapsack problem optimally (see Figure~\ref{fig:NUAlgorithm} for an illustration).

\begin{algorithm}[H]
\caption{Nemhauser-Ullmann algorithm}
\begin{algorithmic}[1]
\State $\calP_0 := \{0^n\}$;
  \For{$i=1,\ldots,n$}
	  \State $\calQ_i := \calP_{i-1} \cup \calP_{i-1}^{+i}$;\label{line:NUAlgo3}
	  \State $\calP_i:=\{x\in \calQ_i\mid\not\exists y\in\calQ_i \colon y \text{ dominates } x\}$;\label{line:NUAlgo4}
  \EndFor
\State\Return $x^{\star} := \argmax_{x\in\calP_n} \{p\tra x \mid w\tra x \le W\}$;\label{line:NUAlgo5}
\end{algorithmic}
\label{algorithm:NUAlgorithm}
\end{algorithm}

In Line~\ref{line:NUAlgo4} a tie-breaking is assumed so that~$\calP_i$ does never contain two solutions with identical profits and weights.

Observe that all steps of the Nemhauser-Ullmann algorithm except for Line~\ref{line:NUAlgo5} are independent of the capacity~$W$. In order to speed up the algorithm, one could remove solutions with weights larger than~$W$ already from~$\calQ_i$ in Line~\ref{line:NUAlgo3}.

\begin{figure}[t]
\centering
  \epsfig{file=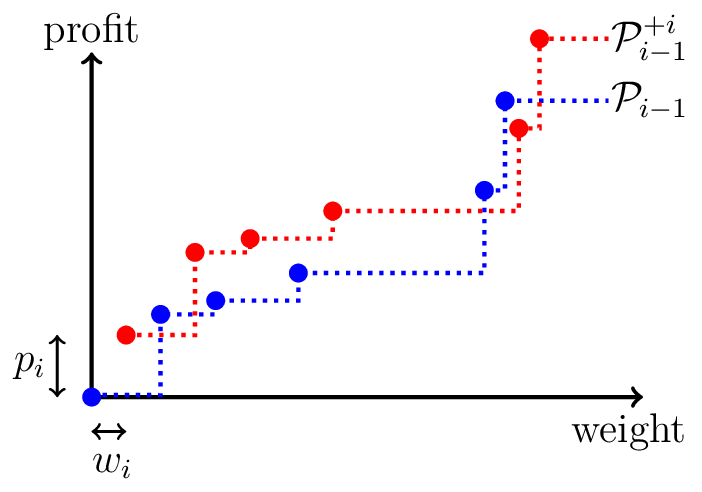,width=.48\textwidth}
 \hspace{0.3cm}
  \epsfig{file=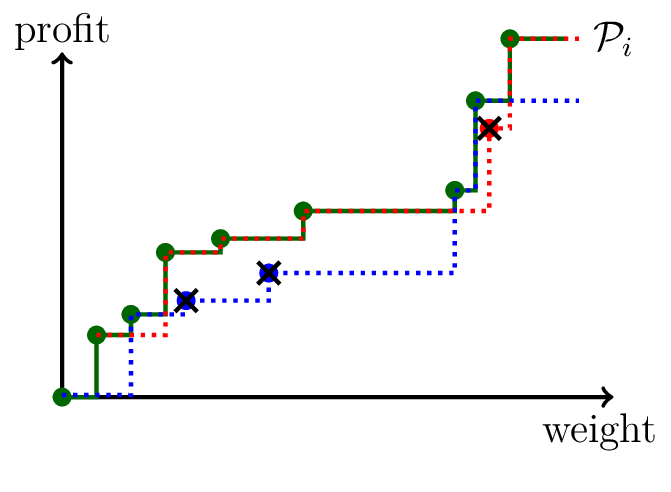,width=.48\textwidth}
\caption{Illustration of one iteration of the for-loop of the Nemhauser-Ullmann algorithm: The set~$\calP_{i-1}^{+i}$ is a copy of the set~$\calP_{i-1}$ that is shifted by~$(w_i,p_i)$. The set~$\calP_{i}$ is obtained by removing dominated solutions.}
\label{fig:NUAlgorithm}
\end{figure}

We analyze the running time of the Nemhauser-Ullmann algorithm using the model of a unit-cost RAM.
In this model, arithmetic operations like adding and comparing two numbers can be performed in constant time regardless
of their bit-lengths. We use this model for the sake of simplicity and in order to keep
the focus on the important details of the running time analysis.

\begin{theorem}\label{thm:NemhauserUllmannRunningTime}
The Nemhauser-Ullmann algorithm solves the knapsack problem optimally.
There exists an implementation with running time~$\Theta(\sum_{i=0}^{n-1}|\calP_{i}|)$.
\end{theorem}
\begin{proof}
The correctness of the algorithm follows immediately from the previous discussion.
In order to achieve the claimed running time, we do not compute the sets~$\calP_i$ explicitly, but only
the values of the solutions in these sets. That is, instead of~$\calP_i$ only the
set~$\val(\calP_i):=\{(p\tra x,w\tra x)\mid x\in\calP_i\}$ is computed. Analogously to the computation of~$\calP_i$, one can
compute~$\val(\calP_i)$ easily if~$\val(\calP_{i-1})$ is known.
If we store for each element of~$\val(\calP_i)$ a pointer to the element of~$\val(\calP_{i-1})$ from which
it originates, then in Step 5 the solution~$x^{\star}$ can be efficiently reconstructed from the sets~$\val(\calP_i)$ and
these pointers.

The running times of Steps~1 and~5 are~$O(1)$ and~$O(n+|\calP|)$, respectively, where the term~$n$
accounts for the running time of reconstructing the solution~$x^{\star}$ once its value~$(p\tra x^{\star},w\tra x^{\star})$ is determined.
In every iteration~$i$ of the for-loop, the running time of Step~3 to compute~$\val(\calQ_i)$ is~$\Theta(|\calP_{i-1}|)$
because on a unit-cost RAM the set~$\val(\calP_{i-1}^{+i})$ can be computed in time~$\Theta(|\calP_{i-1}|)$ from
the set~$\val(\calP_{i-1})$.

In a straightforward implementation, the running time of Step~4
is~$\Theta(|\calQ_i|^2)=\Theta(|\calP_{i-1}|^2)$ because
we need to compare every pair of values from~$\val(\calQ_i)$ and each comparison takes time~$O(1)$.
Step~4 can be implemented more efficiently. For this, we store the values in each
set~$\val(\calP_i)$ sorted in non-decreasing order of weights.
If~$\val(\calP_{i-1})$ is sorted in this way, then, without any additional computational effort, the computation of the set~$\val(\calQ_i)$ in Step~3 can be implemented such that~$\val(\calQ_i)$ is also sorted:
The sorted set~$\val(\calP_{i-1}^{+i})$ can be computed in time~$\Theta(|\calP_{i-1}|)$. Then, in order
to compute the set~$\val(\calQ_i)$, only the two sorted sets~$\val(\calP_{i-1})$ and~$\val(\calP_{i-1}^{+i})$ need
to be merged in time~$\Theta(|\calP_{i-1}|)$. If the set~$\val(\calQ_i)$ is sorted, Step~4 can be implemented
to run in time~$\Theta(|\calQ_i|)$
as a sweep algorithm going once through~$\val(\calQ_i)$ in non-decreasing order of weights (see Exercise~\ref{exer:NUAlgo}). 
\end{proof}

The above theorem ensures that the Nemhauser-Ullmann algorithm solves the knapsack problem efficiently if all Pareto sets~$\calP_i$ have polynomial size.\footnote{Let us remark that the sizes of the Pareto sets are in general not monotone and there are instances where~$|\calP_{i+1}| < |\calP_{i}|$ for some~$i$. Hence, it does not suffice if only the Pareto set~$\calP=\calP_n$ has polynomial size. However, we are not aware of any class of instances where~$|\calP_n|$ is polynomially bounded while~$|\calP_i|$ is superpolynomial for some~$i$.} As the knapsack problem is NP-hard, it is not surprising that there are instances with exponentially many Pareto-optimal solutions. If one sets~$p_i=w_i=2^i$ for each item~$i\in\{1,\ldots,n\}$ then even every solution from~$\{0,1\}^n$ is Pareto-optimal.

\subsection{Shortest Path Problem}\label{subsec:ShortestPathProblem}

Shortest path problems often come naturally with multiple objectives. Think for example of automotive  navigation systems in which one can usually choose between the shortest, cheapest, and fastest route. Let us consider the \emph{bicriteria single-source shortest path problem}. An instance of this problem is described by a directed graph~$G=(V,E)$ with costs $c:E\to\RR_{>0}$, weights $w:E\to\RR_{>0}$, and a source vertex~$s\in V$. The goal is to compute for each~$v\in V$ the set~$\calP^v$ of Pareto-optimal $s$-$v$-paths according to the following definition.

\begin{definition}
For an $s$-$v$-path~$P$ we denote by $w(P)=\sum_{e\in P}w(e)$ its weight and by $c(P)=\sum_{e\in P}c(e)$ its costs. An $s$-$v$-path~$P_1$ \emph{dominates} an $s$-$v$-path~$P_2$ if $w(P_1)\le w(P_2)$ and $c(P_1)\le c(P_2)$, with at least one of these inequalities being strict. An $s$-$v$-path~$P$ is called \emph{Pareto-optimal} if it is not dominated by any other $s$-$v$-path.
\end{definition}

A well-known algorithm for the single-criterion single-source shortest path problem (with only weights but no costs on the edges) is the Bellman-Ford algorithm. It stores a distance label for each vertex which is initially set to infinity for each vertex except the source~$s$ for which it is set to zero. Then it performs a sequence of relax operations on the edges as shown in the following pseudocode.

\begin{algorithm}[H]
\caption{Bellman-Ford algorithm}
\begin{algorithmic}[1]
\State $\dist(s)=0$;
\For{$v\in V\setminus\{s\}$} $\dist(v)=\infty$; \EndFor
  \For{$i=1,\ldots,|V|-1$}
    \For{each $(u,v)\in E$} 
      \State $\text{\sc relax}(u,v)$; 
    \EndFor
  \EndFor
\Procedure{relax}{$u,v$}
  \If{$\dist(v) > \dist(u) + w(u,v)$}
    \State $\dist(v) := \dist(u) + w(u,v)$;
  \EndIf
\EndProcedure
\end{algorithmic}
\end{algorithm}

It can be shown that after termination the distance label~$\dist(v)$ of each vertex~$v$ equals the length of the shortest $s$-$v$-path in~$G$. By standard methods one can adapt the algorithm so that for each vertex the actual shortest $s$-$v$-path is computed. One can also easily adapt this algorithm to the bicriteria shortest path problem if one replaces each distance label~$\dist(v)$ by a list~$L_v$ of $s$-$v$-paths. Initially~$L_s$ contains only the trivial path of length~$0$ from~$s$ to~$s$ and all other lists~$L_v$ are empty. In every relax operation for an edge~$(u,v)$ a new set $L_u^{+(u,v)}$ is obtained from~$L_u$ by appending the edge~$(u,v)$ to each path from~$L_u$. Then the paths from~$L_u^{+(u,v)}$ are added to~$L_v$. Finally~$L_v$ is cleaned up by removing all paths from~$L_v$ that are dominated by other paths from this list. This is shown in the following pseudocode.

\begin{algorithm}[H]
\caption{Bicriteria Bellman-Ford algorithm}
\begin{algorithmic}[1]
\State $L_s=\{\text{path of length~0 from~$s$ to~$s$}\}$;
\For{$v\in V\setminus\{s\}$} $L_v=\emptyset$; \EndFor
  \For{$i=1,\ldots,|V|-1$}
    \For{each $(u,v)\in E$} 
      \State $\text{\sc relax}(u,v)$; 
    \EndFor
  \EndFor
\Procedure{relax}{$u,v$}
    \State Obtain $L_u^{+(u,v)}$ from~$L_u$ by appending the edge~$(u,v)$ to each path from~$L_u$.
    \State $L_v := L_v \cup L_u^{+(u,v)}$;
    \State Remove dominated paths from~$L_v$.\label{line:BellmanFordRemoveDominated}
\EndProcedure
\end{algorithmic}
\label{algorithm:BFAlgorithm}
\end{algorithm}

Analogously to the Nemhauser-Ullmann algorithm, the running time of the Bicriteria Bellman-Ford algorithm depends crucially on the sizes of the lists~$L_v$ that appear throughout the algorithm. We have to look at the algorithm in slightly more detail to give an upper bound on its running time. The algorithm performs $M:=(|V|-1)\cdot|E|$ relax operations, which we denote by~$R_1,\ldots,R_{M}$. For a relax operation~$R_k$ that relaxes the edge~$(u,v)$, we define $u(R_k)=u$ and $v(R_k)=v$. Let $k\in[M]$ and consider the first~$k$ relax operations. These define for every vertex~$v\in V$ a set~$S_v^k$ of $s$-$v$-paths that can be discovered by the first~$k$ relax operations. To be more precise, $S_v^k$ contains exactly those $s$-$v$-paths that appear as a subsequence in $(u(R_1),v(R_1)),\ldots,(u(R_k),v(R_k))$. In the single-criterion version, after~$k$ relax operations the distance label~$\dist(v)$ contains the length of the shortest path in~$S_v^k$. In the bicriteria version, the list~$L_v$ contains after $k$ relax operations all paths from~$S_v^k$ that are Pareto-optimal within this set (i.e. that are not dominated by other paths from this set). We will denote the list~$L_v$ after $k$ relax operations by $L_v^k$ in the following. 

\begin{theorem}\label{thm:BellmannFordRunningTime}
After termination of the Bicriteria Bellman-Ford algorithm the list~$L_v$ equals for every vertex~$v\in V$ the set of Pareto-optimal $s$-$v$-paths. There exists an implementation with running time~$\Theta\left(\sum_{k=1}^{M}\left(|L_{u(R_k)}^{k-1}|+|L_{v(R_k)}^{k-1}|\right)\right)$.
\end{theorem}
\begin{proof}
The correctness of the algorithm follows by an inductive argument along the lines of the analysis of the single-criterion version (see Exercise~\ref{exer:BellmanFord}). The analysis of the running time is similar to the proof of Theorem~\ref{thm:NemhauserUllmannRunningTime}. The dominating factor is the time to remove dominated paths from~$L_v$ in Line~\ref{line:BellmanFordRemoveDominated} of the pseudocode. A naive implementation has running time $\Theta(|L_{u(R_k)}^{k-1}|\cdot |L_{v(R_k)}^{k-1}|)$ for the $k$th relax operation, while the running time $\Theta(|L_{u(R_k)}^{k-1}| + |L_{v(R_k)}^{k-1}|)$ can be achieved by sweeping through the lists when they are sorted in non-decreasing order of weight.
\end{proof}

While in applications where the bicriteria shortest path problem occurs, it has been observed that the number of Pareto-optimal solutions is usually not very large, one can easily construct instances of the bicriteria shortest path problem in which the number of Pareto-optimal paths is exponential in the size of the graph (see Exercise~\ref{exer:NumberPareto}).

The reader might wonder why we adapted the Bellman-Ford algorithm and not Dijkstra's algorithm to the bicriteria single-source shortest path problem. Indeed there is a generalization of Dijkstra's algorithm to the bicriteria shortest path problem due to \citet{Hansen79}, which also performs a sequence of operations similar to the relax operations of the Bellman-Ford algorithm. However, in contrast to the Bellman-Ford algorithm the sequence of relax operations is not fixed beforehand but it depends on the actual costs and weights of the edges. For this reason, it is not clear how to analyze the expected running time and in particular the analysis that we present in Section~\ref{sec:NumberPOSolutions} does not apply to the generalization of Dijkstra's algorithm.

\subsection{Multiple Objectives and Other Optimization Problems}

For the sake of simplicity, we have discussed only problems with two objectives above. However, one can easily adapt the definition of Pareto-optimal solutions and both presented algorithms to more than two objectives. Consider the multidimensional knapsack problem, a version of the knapsack problem in which every item still has a single profit but instead of a single weight, it has a weight vector from~$\RRp^{d-1}$ for some~$d\ge 2$, and also the capacity is a vector from~$\RRp^{d-1}$. This problem gives rise to a multiobjective optimization problem with~$d$ objectives: maximize the profit~$p\tra x$ and minimize for each~$i\in[d-1]$ the $i$th weight~$(w^{(i)})\tra x$. Similarly it is often natural to consider multiobjective shortest path problems with more than two objectives.

In order to compute the Pareto set of an instance of the multidimensional knapsack problem or the multiobjective shortest path problem, no modification to the pseudocode of the Nemhauser-Ullmann algorithm (Algorithm~\ref{algorithm:NUAlgorithm}) and the Bicriteria Bellman-Ford algorithm (Algorithm~\ref{algorithm:BFAlgorithm}) are necessary. However, the implementation and the analysis of the running time have to be adapted. The crucial difference is that the removal of dominated solutions from~$\calQ_i$ and~$L_v$ cannot be implemented in time linear in the sizes of these sets anymore because the sweeping approach, which assumes that the solutions are sorted with respect to one of the objectives, fails for more than two objectives. If one uses the naive implementation, which pairwisely compares the solutions, then the running times of the algorithms become $\Theta(\sum_{i=0}^{n-1}|\calP_{i}|^2)$ and $O(\sum_{i=1}^{M}|L_{u(R_i)}^{i-1}|\cdot |L_{v(R_i)}^{i-1}|)$, respectively.

Asymptotically one can do better by using known algorithms for the maximum vector problem to filter out the dominated solutions. In this problem a set of~$m$ vectors in~$\RR^k$ is given and one wants to compute the set of Pareto-optimal vectors among them. The fastest known algorithm for this problem is due to \citet{KungLP75}. It relies on divide and conquer and its running time is $O(m\log^{k-2}m)$. For~$d$ objectives this yields running times of
$
   \Theta\left(\sum_{i=0}^{n-1}|\calP_{i}|\log^{d-2}(|\calP_{i}|)\right)
$
and
\[
     O\left(\sum_{i=1}^{M}(|L_{u(R_i)}^{i-1}|+ |L_{v(R_i)}^{i-1}|)\cdot\log^{d-2}(|L_{u(R_i)}^{i-1}|+ |L_{v(R_i)}^{i-1}|)\right)
\]
for the Nemhauser-Ullmann algorithm and the Bellman-Ford algorithm, respectively.

The Nemhauser-Ullmann algorithm and the Bicriteria Bellman-Ford algorithm are only two examples of many algorithms in the literature for computing Pareto sets of various multiobjective optimization problems. Similar algorithms exist, for example, for the multiobjective network flow problem. As a rule of thumb, algorithms that solve the single-criterion version of an optimization problem by dynamic programming can usually be adapted to compute the Pareto set of the multiobjective version.

On the other hand, there are also problems for which it is unknown if there exist algorithms that compute the Pareto set in time polynomial in its size and the sizes of the Pareto sets of appropriate subproblems. The multiobjective spanning tree problem is one such example, where the best known way to compute the Pareto set is essentially to first compute the set of all spanning trees and then to remove the dominated ones. An even stronger requirement is that of an efficient output-sensitive algorithm, which computes the Pareto set in time polynomial in its size and the input size. \citet{BoeklerEMM17} show that such an algorithm exists for the multiobjective version of the minimum-cut problem and that no such algorithm exists for the bicriteria shortest path problem, unless P$=$NP. For many other multiobjective problems, including the knapsack problem and the multiobjective spanning tree problem, it is an open question whether efficient output-sensitive algorithms exist.

\subsection{Approximate Pareto Curves}

For virtually every multiobjective optimization problem the number of Pareto-optimal solutions can be exponential in the worst case. One way of coping with this problem is to relax the requirement of finding the complete Pareto set. A solution $x$ is \emph{$\varepsilon$-dominated} by a solution $y$ if $y$ is worse than~$x$ by at most a factor of~$1+\varepsilon$ in each objective (i.e., $w(y)/w(x)\leq 1+\varepsilon$ for each criterion~$w$ that is to be minimized and $p(x)/p(y)\leq 1+\varepsilon$ for each criterion~$p$ that is to be maximized). We say that $\calP_{\varepsilon}$ is an \emph{$\varepsilon$-approximation of a Pareto set} $\calP$ if for any solution in $\calP$, there is a solution in $\calP_{\varepsilon}$ that $\varepsilon$-dominates it.

In his pioneering work, \citet{Hansen80} presents an approximation scheme for computing $\varepsilon$-approximate Pareto sets of the bicriteria shortest path problem. \citet{PapadimitriouY00} show that for any instance of a multiobjective optimization problem, there is an $\varepsilon$-approximation of the Pareto set whose size is polynomial in the input size and $1/\varepsilon$ but exponential in the number of objectives. Furthermore, they define the \emph{gap version} of a multiobjective optimization problem with $d$ objectives as follows: given an instance of the problem and a vector $b\in\RR^d$, either return a solution whose objective vector dominates $b$ or report (correctly) that there does not exist any solution whose objective vector is better than $b$ by more than a $(1+\varepsilon)$ factor in all objectives. They show that an FPTAS for approximating the Pareto set of a multiobjective optimization problem exists if and only if the gap version of the problem can be solved in polynomial time. In particular, this implies that if the exact single-criterion version of a problem (i.e., the question ``Is there a solution with weight exactly $x$?'') can be solved in pseudopolynomial time, then its multiobjective version admits an FPTAS for approximating the Pareto set. This is the case, for example, for the spanning tree problem, the all-pair shortest path problem, and the perfect matching problem.

\citet{VassilvitskiiY04} show how to compute $\varepsilon$-approximate Pareto sets whose size is at most three times as large as the smallest such set for bicriteria problems whose gap versions can be solved in polynomial time. \citet{DiakonikolasY07} improve this factor to two and show that this is the best possible that can be achieved in polynomial time, unless P$=$NP.

\section{Number of Pareto-optimal solutions}\label{sec:NumberPOSolutions}

Both for the knapsack problem and the bicriteria shortest path problem, the number of Pareto-optimal solutions increases only moderately with the input size in applications. This is in contrast to the exponential worst-case behavior (see Exercise~\ref{exer:NumberPareto}). To explain this discrepancy, we will analyze the number of Pareto-optimal solutions in the framework of smoothed analysis. First we will focus on the knapsack problem but we will see afterwards that the proven bound also holds for a much larger class of problems including the bicriteria shortest path problem and many other natural bicriteria optimization problems. We will then also briefly discuss known results for problems with more than two objectives.

\subsection{Knapsack Problem}

Let us consider the knapsack problem. In a worst-case analysis the adversary is allowed to choose the profits $p_1,\ldots,p_n$ and the weights $w_1,\ldots,w_n$ exactly (he can also choose the capacity but the number of Pareto-optimal solutions is independent of this). This makes him very powerful and makes it possible to choose an instance in which every solution is Pareto-optimal. In order to limit the power of the adversary to construct such artificial instances that do not resemble typical inputs, we add some randomness to his decisions.

Let~$\phi\ge 1$ be a parameter. In the following analysis, we assume that the adversary can still determine the profits exactly while for each weight he can only choose an interval of length~$1/\phi$ from which it is chosen uniformly at random independently of the other weights. This means that the adversary can specify each weight only with a precision of~$1/\phi$. We normalize the weights and restrict the adversary to intervals that are subsets of~$[0,1]$. This normalization is necessary to ensure that the effect of the noise cannot be ruled out by scaling all weights in the input by some large number.

Observe that the parameter~$\phi$ measures the strength of the adversary. If $\phi=1$ then all weights are chosen uniformly at random from~$[0,1]$, which resembles an average-case analysis. On the other hand, in the limit for $\phi\to\infty$ the adversary can determine the weights (almost) exactly and the model approaches a classical worst-case analysis. Hence, it is not surprising that the bound that we will prove for the expected number of Pareto-optimal solutions grows with~$\phi$. However, we will see that it grows only polynomially with~$n$ and~$\phi$, which implies that already a small amount of random noise suffices to rule out the worst case and to obtain a benign instance in expectation.

\begin{theorem} \label{thm:UpperPareto}
Consider an instance~$\calI$ of the knapsack problem with arbitrary profits~$p_1,\ldots,p_n\in\RRp$ in which every weight~$w_i$ is chosen uniformly at random from an arbitrary interval~$A_i\subseteq[0,1]$ of length~$1/\phi$ independently of the other weights. Then the expected number of Pareto-optimal solutions in~$\calI$ is bounded from above by $n^2\phi+1$.
\end{theorem}

The proof of Theorem~\ref{thm:UpperPareto}, which we present in detail below, can be summarized as follows. Since all weights take values between~$0$ and~$1$, all solutions have weights between~$0$ and~$n$. We divide the interval~$[0,n]$ uniformly into a large number~$k$ of subintervals of length~$n/k$ each. For large enough~$k$ it is unlikely that there exist two Pareto-optimal solutions whose weights lie in the same subinterval because the weights are continuous random variables. Assuming that this does not happen, the number of Pareto-optimal solutions equals the number of subintervals that contain a Pareto-optimal solution. The most important and non-trivial step is then to bound, for each subinterval, the probability that it contains a Pareto-optimal solution. Once we have proven an upper bound for this, the theorem follows by summing up this upper bound over all subintervals due to linearity of expectation.

Before we prove the theorem, we state one simple but crucial property of the random variables that we consider.
\begin{lemma}\label{lemma:epsIntervalEasy}
Let~$X$ be a random variable that is chosen uniformly at random from some interval~$A$ of length~$1/\phi$. Furthermore let~$I$ be an interval of length~$\eps$. Then $\Pr[X\in I] \le \phi\eps$.
\end{lemma}
\begin{proof}
Since~$X$ is chosen uniformly at random from~$A$, we obtain
\[
   \Pr[X\in I] = \frac{|A\cap I|}{|A|} \le \frac{|I|}{|A|} \le 
   \frac{\eps}{1/\phi} = \phi\eps.\qedhere
\]
\end{proof}

\begin{proof}[Proof of Theorem~\ref{thm:UpperPareto}]
Every solution~$x\in\{0,1\}^n$ has a weight~$w\tra x$ in the interval~$[0,n]$ because each weight~$w_i$ lies in~$[0,1]$. We partition the interval~$(0,n]$ uniformly into $k\in\NN$~intervals~$I_0^k,\ldots,I_{k-1}^k$ for some large number~$k$ to be chosen later. Formally, let~$I_i^k=(ni/k,n(i+1)/k]$. We say that the interval~$I_i^k$ is \emph{non-empty} if there exists a Pareto-optimal solution~$x\in\calP$ with~$w\tra x\in I_i^k$.

We denote by~$X^k$ the number of non-empty intervals~$I_i^k$ plus one. The term~$+1$ accounts for the solution~$0^n$, which is always Pareto-optimal and does not belong to any interval~$I_i^k$. Nevertheless, the variable~$X^k$ can be much smaller than~$|\calP|$ because many Pareto-optimal solutions could lie in the same interval~$I_i^k$. We will ensure that every interval~$I_i^k$ contains at most one Pareto-optimal solution with high probability by choosing~$k$ sufficiently large. Then, with high probability,~$|\calP|=X^k$.

In the following, we make this argument more formal. For~$k\in\NN$, let~$\calF_{k}$ denote the event that there exist two different solutions~$x,y\in\{0,1\}^n$ with~$|w\tra x - w\tra y|\le n/k$. Since each interval~$I_i^k$ has length~$n/k$, every interval~$I_i^k$ contains at most one Pareto-optimal solution if~$\calF_k$ does not occur.
\begin{lemma}\label{lemma:Fk}
For every~$k\in\NN$, $\Pr\left[\calF_{k}\right] \le \frac{2^{2n+1}n\phi}{k}$.
\end{lemma}
\begin{proof}
There are~$2^{n}$ choices for~$x$ and~$y$ each. We prove the lemma by a union bound over all these choices. Let~$x,y\in\{0,1\}^n$ with~$x\neq y$ be fixed. Then there exists an index~$i$ with~$x_i\neq y_i$. Assume without loss of generality that~$x_i=0$ and~$y_i=1$. We use the principle of deferred decisions and assume that all weights~$w_j$ except for~$w_i$ are already fixed. Then $w\tra x - w\tra y = \alpha - w_i$ for some constant~$\alpha$ that depends on~$x$ and~$y$ and the fixed profits~$w_j$. It holds that
\begin{align*}
   \Pr\left[|w\tra x - w\tra y|\le \frac{n}{k}\right] 
   & \le \sup_{\alpha\in\RR}\Pr_{w_i}\left[|\alpha-w_i|\le \frac{n}{k}\right] \\
   & = \sup_{\alpha\in\RR}\Pr_{w_i}\left[w_i\in\bigg[\alpha-\frac{n}{k},\alpha+\frac{n}{k}\bigg]\right] \le \frac{2n\phi}{k},
\end{align*}
where the last inequality follows from Lemma~\ref{lemma:epsIntervalEasy}\footnote{Formally, we condition on the outcome of the~$w_j$ with~$j\neq i$. This outcome determines the value of~$\alpha$. Then we apply the law of total probability, but instead of integrating over all possible outcomes of the~$w_j$ with~$j\neq i$, we derive an upper bound by looking only at the worst choice for~$\alpha$.}. Now a union bound over all choices for~$x$ and~$y$ concludes the proof.
\end{proof}

The most non-trivial part in the analysis is the following lemma, which states for an arbitrary interval an upper bound for the probability that it contains a Pareto-optimal solution. We defer the proof of this lemma to the end of this section.

\begin{lemma}\label{lemma:ProbPOInterval}
For every~$t\ge 0$ and every~$\eps>0$,
\[
  \Pr[\exists x\in\calP\mid w\tra x \in (t,t+\eps]] \le n\phi\eps.
\]
\end{lemma}

The following lemma is the main building block in the proof of the theorem.
\begin{lemma}\label{lemma:ExXk}
For every~$k\in\NN$, $\Ex{X^k}\le n^2\phi+1$.
\end{lemma}
\begin{proof}
Let~$X^k_i$ denote a random variable that is~$1$ if the interval~$I_i^k$ is non-empty and~$0$ otherwise. Then
\[
   X^k = 1+\sum_{i=0}^{k-1} X^k_i
\]
and by linearity of expectation
\begin{equation}\label{eqn:Xk1}
   \Ex{X^k} = \Ex{1+\sum_{i=0}^{k-1} X^k_i} = 1+\sum_{i=0}^{k-1}\Ex{X^k_i}.
\end{equation}
Since~$X^k_i$ is a $0$-$1$-random variable, its expected value can be written as
\begin{equation}\label{eqn:Xk2}
  \Ex{X^k_i} = \Pr[X^k_i=1] = \Pr[\exists x\in\calP\mid w\tra x \in I_i^k].
\end{equation}

Using that each interval~$I_i^k$ has length~$n/k$, Lemma~\ref{lemma:ProbPOInterval} and~\eqref{eqn:Xk2} imply
\[
    \Ex{X^k_i} \le \frac{n^2\phi}{k}.
\]
Together with~\eqref{eqn:Xk1} this implies
\[
  \Ex{X^k} = 1+\sum_{i=0}^{k-1} \Ex{X^k_i} \le 1+k\cdot \frac{n^2\phi}{k} = n^2\phi+1.\qedhere
\]
\end{proof}

With the help of Lemmas~\ref{lemma:Fk} and~\ref{lemma:ExXk}, we can finish the proof of the theorem as follows:
\begin{align}
  \Ex{|\calP|} & = \sum_{i=1}^{2^n} \big(i \cdot \Pr[|\calP|=i]\big)\notag\\
               & = \sum_{i=1}^{2^n} \big(i \cdot \Pr[|\calP|=i\wedge \calF_{k}] + i \cdot \Pr[|\calP|=i\wedge \neg\calF_{k}]\big)\notag\\
               & \stackrel{(1)}{=} \sum_{i=1}^{2^n} \big(i \cdot \Pr[\calF_{k}]\cdot \Pr[|\calP|=i\mid \calF_{k}]\big) + \sum_{i=1}^{2^n}\big(i\cdot \Pr[X^k=i\wedge \neg\calF_{k}]\big)\notag\\
			   & \le \Pr[\calF_{k}]\cdot \sum_{i=1}^{2^n} \big(i \cdot  \Pr[|\calP|=i\mid \calF_{k}]\big) + \sum_{i=1}^{2^n}\big(i\cdot \Pr[X^k=i]\big)\notag\\
			   & \stackrel{(2)}{\le} \frac{2^{2n+1}n\phi}{k}\cdot \sum_{i=1}^{2^n} \big(2^n \cdot \Pr[|\calP|=i\mid \calF_{k}] \big) + \Ex{X^k}\notag\\
			   & \stackrel{(3)}{\le} \frac{2^{3n+1}n\phi}{k} + n^2\phi+1.\label{eqn:NumberPareto}
\end{align}
Let us comment on some of the steps in the previous calculation.
\begin{itemize}
\item The upper bound~$2^n$ on the indices of the sums follows because $|\calP|$ can never exceed the total number of solutions, which is~$2^n$.
\item The rewriting of the first term in (1) follows from the definition of the conditional probability
          and the rewriting of the second term follows because~$X^k=|\calP|$ when the event~$\neg\calF_k$ occurs.
\item (2) follows from Lemma~\ref{lemma:Fk} and the definition of the expected value.
\item (3) follows from the identity~$\sum_{i=1}^{2^n}\Pr[|\calP|=i\mid \calF_{k}]=1$ and Lemma~\ref{lemma:ExXk}.
\end{itemize}

Since~\eqref{eqn:NumberPareto} holds for every~$k\in\NN$, it must be~$\Ex{|\calP|} \le n^2\phi+1$.
\end{proof}
It only remains to prove Lemma~\ref{lemma:ProbPOInterval}. An easy way to derive an upper bound for the probability that there exists a Pareto-optimal solution in the interval~$(t,t+\eps]$ is to apply a union bound over all solutions. Since there is an exponential number of solutions, this does not lead to a useful bound. The key improvement in the proof of Lemma~\ref{lemma:ProbPOInterval} is to apply the union bound only over the $n$ dimensions.

\begin{proof}[Proof of Lemma~\ref{lemma:ProbPOInterval}]

\begin{figure}
\centering
  \epsfig{file=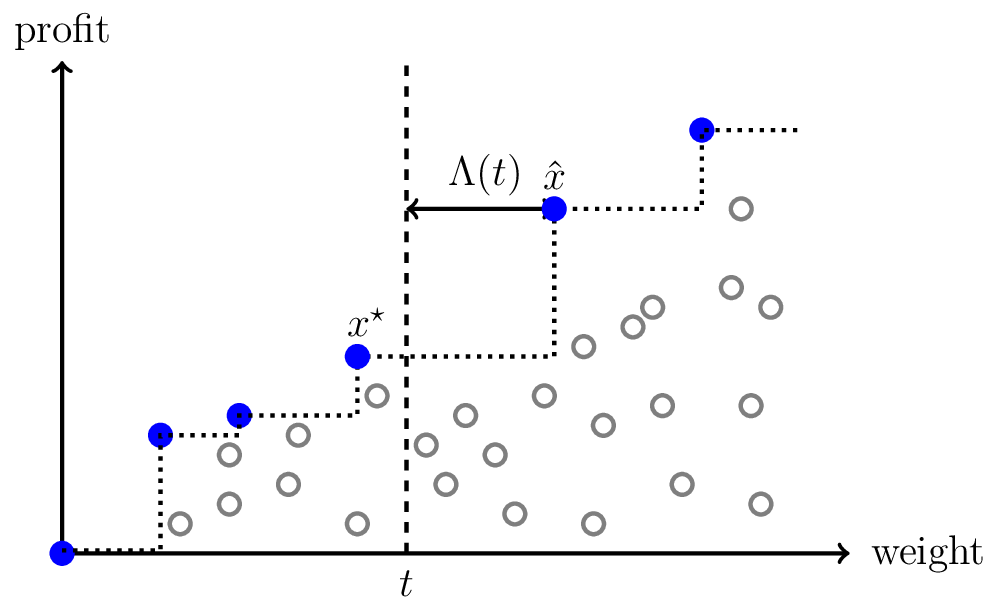,width=.6\textwidth}
  \caption{Definitions of the winner $x^{\star}$, the loser $\hat{x}$, and the
  random variable $\Lambda(t)$.}
  \label{fig:WinnerLoser}
\end{figure}
Fix~$t\ge 0$ and~$\eps>0$. First of all we define a random variable~$\Lambda(t)$.
In order to define~$\Lambda(t)$, we define the \emph{winner}~$x^{\star}$ to be
the most valuable solution satisfying~$w\tra x \le t$, i.e.,
\[
  x^{\star} = \argmax\{p\tra x\mid x\in\{0,1\}^n\text{ and } w\tra x\le t\}.
\]
For~$t\ge0$, such a solution~$x^{\star}$ must always exist. We say that a
solution~$x$ is a \emph{loser} if it has a higher profit than~$x^{\star}$. By the choice of~$x^{\star}$, losers
do not satisfy the constraint~$w\tra x\le t$ (hence their name). We denote by~$\hat{x}$
the loser with the smallest weight (see Figure~\ref{fig:WinnerLoser}), i.e.,
\[
  \hat{x} = \argmin\{w\tra x\mid x\in\{0,1\}^n\text{ and } p\tra x> p\tra x^{\star}\}.
\]
If there does not exist a solution~$x$ with~$p\tra x>p\tra x^{\star}$, then~$\hat{x}$ is undefined,
which we denote by~$\hat{x}=\perp$. Based on~$\hat{x}$, we define the random
variable~$\Lambda(t)$ as
\[
  \Lambda(t) = \begin{cases}
                 w\tra\hat{x}-t & \text{if } \hat{x}\neq\,\perp,\\
                 \infty & \text{if~$\hat{x}=\,\perp$.}
               \end{cases}
\]

The random variable~$\Lambda(t)$ satisfies the following equivalence:
\begin{equation}\label{eqn:Equivalence}
   \Lambda(t) \le \varepsilon \iff \exists x\in\calP\colon
  w\tra x\in(t,t+\varepsilon].
\end{equation} 
To see this, assume that there exists a Pareto-optimal solution whose weight lies in~$(t,t+\varepsilon]$,
and let~$y$ denote the Pareto-optimal solution
with the smallest weight in~$(t,t+\varepsilon]$. Then~$y=\hat{x}$ and
hence~$\Lambda(t)=w\tra\hat{x}-t \in (0,\varepsilon]$. Conversely, if~$\Lambda(t)\le\varepsilon$,
then~$\hat{x}$ must be a Pareto-optimal solution whose weight lies in the interval~$(t,t+\varepsilon]$.
Together this yields Equivalence~(\ref{eqn:Equivalence}).
Hence,
\begin{equation}\label{eqn:ExPareto}
  \Pr[\exists x\in\calP\mid w\tra x \in (t,t+\eps]]
   = \Pr[\Lambda(t)\le \eps].
\end{equation}

It only remains to bound the probability that $\Lambda(t)$ does not
exceed $\varepsilon$. In order to analyze this probability, we define a
set of auxiliary random variables~$\Lambda^1(t),\ldots,\Lambda^n(t)$ such that $\Lambda(t)$ is guaranteed
to always take a value also taken by at least one of the auxiliary random
variables. Then we analyze the auxiliary random variables and use a
union bound to conclude the desired bound for $\Lambda(t)$. 

Let $i\in[n]$ be fixed.
The random variable $\Lambda^i(t)$ is defined similarly to $\Lambda(t)$, but
only solutions that do not contain item~$i$ are eligible as
winners and only solutions that contain item~$i$ are eligible as losers. We make this more formal in the following.
For $j\in\{0,1\}$, we define
\[
\calS^{x_i=j}=\{x\in\{0,1\}^n\mid x_i=j\},
\]
and we define $x^{\star,i}$ to be
\[
  x^{\star,i} = \argmax\{p\tra x\mid x\in\calS^{x_i=0}\text{ and } w\tra x\le t\}.
\]
That is, $x^{\star,i}$ is the winner among the solutions that do not
contain item~$i$. We restrict our attention to losers that
contain item~$i$ and define
\[
  \hat{x}^i = \argmin\{w\tra x\mid x\in\calS^{x_i=1}\text{ and } p\tra x>
  p\tra x^{\star,i}\}.
\]
If there does not exist a solution $x\in\calS^{x_i=1}$ with $p\tra
x>p\tra x^{\star,i}$, then $\hat{x}^i$ is undefined, i.e.,
$\hat{x}^i=\,\perp$. Based on $\hat{x}^i$, we define the random variable
$\Lambda^i(t)$ as
\[
  \Lambda^i(t) = \begin{cases}
                 w\tra\hat{x}^i-t & \text{if } \hat{x}^i\neq\,\perp,\\
                 \infty & \text{if $\hat{x}^i=\,\perp$.}
               \end{cases}
\]

\begin{lemma}\label{lemma:Lambda}
For every choice of profits and weights, either $\Lambda(t)=\,\infty$ or
there exists an index $i\in[n]$ such that $\Lambda(t)=\Lambda^i(t)$.
\end{lemma}
\begin{proof}
Assume that $\Lambda(t)\neq\,\infty$. Then there exist a winner $x^{\star}$ and
a loser $\hat{x}$. Since $x^{\star}\neq\hat{x}$, there must be an index
$i\in[n]$ with $x^{\star}_i\neq\hat{x}_i$. Since all weights are non-negative and
$w\tra x^{\star} < w\tra\hat{x}$, there must even be an index $i\in[n]$
with $x^{\star}_i=0$ and $\hat{x}_i=1$. We claim that for this index~$i$,
$\Lambda(t)=\Lambda^i(t)$. In order to see this, we first observe that
$x^{\star}=x^{\star,i}$. This follows because $x^{\star}$ is the solution with the
highest profit among all solutions with weight at most $t$. Since
it belongs to $\calS^{x_i=0}$, it is in particular the solution with the
highest profit among all solutions that do not contain item~$i$ and have weight at most~$t$.
Since $x^{\star}=x^{\star,i}$, by similar
arguments it follows that $\hat{x}=\hat{x}^i$. This directly implies
that $\Lambda(t)=\Lambda^i(t)$.
\end{proof}

\begin{lemma}\label{lemma:Lambda2}
For every $i\in[n]$ and every $\varepsilon\ge 0$,
\[
   \Pr[\Lambda^i(t)\in(0,\varepsilon]] \le \phi\varepsilon.
\]
\end{lemma}
\begin{proof}
In order to prove the lemma, it suffices to exploit the randomness of
the weight~$w_i$. We apply the principle of deferred decisions and assume that all other weights are fixed
arbitrarily. Then the weights of all solutions from $\calS^{x_i=0}$ and
hence also the solution $x^{\star,i}$ are fixed because $w_i$ does not influence the solutions in $\calS^{x_i=0}$ and the profits $p_1,\ldots, p_n$ are fixed. If the solution $x^{\star,i}$
is fixed, then also the set of losers $\calL=\{x\in\calS^{x_i=1}\mid p\tra x > p\tra x^{\star,i}\}$ is fixed.
Since, by definition, all solutions from~$\calL$ contain item~$i$ the identity of
the solution~$\hat{x}^i$ does not depend on~$w_i$. (Of course, the weight~$w\tra\hat{x}^i$ depends on~$w_i$. However, which solution will become~$\hat{x}^i$ is independent of~$w_i$.)
This implies that, given the
fixed values of the weights $w_j$ with $j\neq i$, we can rewrite the
event $\Lambda^i(t)\in(0,\varepsilon]$ as $w\tra \hat{x}^i-t\in(0,\varepsilon]$ for a fixed solution $\hat{x}^i$. For a
constant $\alpha\in\RR$ depending on the fixed values of the weights
$w_j$ with $j\neq i$, we can rewrite this event as
$w_i\in(\alpha,\alpha+\varepsilon]$. By
Lemma~\ref{lemma:epsIntervalEasy}, the probability of this
event is bounded from above by~$\phi\varepsilon$.
\end{proof}

Combining Lemmas~\ref{lemma:Lambda} and~\ref{lemma:Lambda2} yields
\[
  \Pr[\Lambda(t)\le\varepsilon] \le
  \Pr[\exists i\in[n]\colon\Lambda^{i}(t)\in(0,\varepsilon]] \le
  \sum_{i=1}^n\Pr[\Lambda^{i}(t)\in(0,\varepsilon]] \le n\phi \varepsilon.
\]
Together with~\eqref{eqn:ExPareto} this proves the lemma.
\end{proof}

Theorem~\ref{thm:UpperPareto} implies the following result on the running time of the
Nemhauser-Ullmann algorithm.
\begin{corollary}\label{cor:NUAlgo}
Consider an instance~$\calI$ of the knapsack problem with arbitrary profits~$p_1,\ldots,p_n\in\RRp$ in which every weight~$w_i$ is chosen uniformly at random from an arbitrary interval~$A_i\subseteq[0,1]$ of length~$1/\phi$ independently of the other weights. Then the expected running time of the Nemhauser-Ullmann algorithm is~$O(n^3 \phi)$.
\end{corollary}
\begin{proof}
It follows from Theorem~\ref{thm:NemhauserUllmannRunningTime} that
the expected running time of the Nemhauser-Ullmann algorithm is bounded from above by
\[
  O\left(\Ex{\sum_{i=0}^{n-1}|\calP_i|}\right),
\]
where~$\calP_i$ denotes the Pareto set of the restricted instance that consists only of the
first~$i$ items. Using linearity of expectation and Theorem~\ref{thm:UpperPareto}, we obtain that
this term is bounded from above by
\[
  O\left(\sum_{i=0}^{n-1}\Ex{|\calP_i|}\right)
  = O\left(\sum_{i=0}^{n-1}(i^2\phi+1)\right)
  = O(n^3\phi).\qedhere
\]
\end{proof}

The decision to add randomness only to the weights is arbitrary. Of course if both the profits and the weights are chosen independently uniformly at random from intervals of length~$1/\phi$ then the upper bound still applies. With minor modifications, the analysis can also be adapted to the case that only the profits are random while the weights are adversarial.

\subsection{General Model}\label{subsec:GeneralModel}

Theorem~\ref{thm:UpperPareto} can be extended in several ways. First of all, the noise model can be generalized to a much wider class of distributions. In fact the only property that we used about the random weights is Lemma~\ref{lemma:epsIntervalEasy}, which says that the probability to fall into any interval of length~$\eps$ is at most~$\phi\eps$. This is true for every random variable that is described by a probability density function that is bounded from above by~$\phi$. Hence instead of allowing the adversary to choose an interval of length~$1/\phi$ for each weight~$w_i$, we could also allow him to choose a density function~$f_i:[0,1]\to[0,\phi]$ according to which~$w_i$ is chosen independently of the other weights. This includes as a special case the uniform distribution in an interval of length~$1/\phi$ but it also allows different types of random noise. Observe that we have restricted the density functions to~$[0,1]$ to normalize the weights.

In the following we will use the term \emph{$\phi$-perturbed random variable} to refer to a random variable described by a density~$f:\RR\to[0,\phi]$. If we replace all occurrences of Lemma~\ref{lemma:epsIntervalEasy} in the proof of Theorem~\ref{thm:UpperPareto} by the following lemma then Theorem~\ref{thm:UpperPareto} follows also for general $\phi$-perturbed weights from~$[0,1]$.

\begin{lemma}\label{lemma:epsInterval}
Let~$X$ be a $\phi$-perturbed random variable that is described by a density function~$f:[0,1]\to[0,\phi]$. For any interval~$I$ of length~$\eps$, $\Pr[X\in I] \le \phi\eps$.
\end{lemma}
\begin{proof}
The lemma follows by the following simple calculation:
\[
   \Pr[X\in I] = \int_{I} f(x)\,dx \le \int_{I} \phi \,dx = \phi\eps.\qedhere
\]
\end{proof}

Next we state an even more general version of Theorem~\ref{thm:UpperPareto}.
The first generalization compared to Theorem~\ref{thm:UpperPareto} is that an arbitrary set~$\calS\subseteq\{0,1\}^n$ of solutions is given. In the case of the knapsack problem, every vector from~$\{0,1\}^n$ is a solution, i.e.,~$\calS=\{0,1\}^n$. The second generalization is that the adversarial objective function~$p$ does not have to be linear. In fact, it can be an arbitrary function that maps every solution to some real value. The third generalization is that we extend the range of the $\phi$-perturbed weights from~$[0,1]$ to~$[-1,1]$.

\begin{theorem} \label{thm:UpperParetoGeneral}
Let~$\calS\subseteq\{0,1\}^n$ and~$p:\calS\to\RR$ be arbitrary. Let $w_1,\ldots,w_n$ be arbitrary $\phi$-perturbed numbers from the interval~$[-1,1]$. Then the expected number of solutions~$x\in\calS$ that are Pareto-optimal with respect to the objective functions~$p(x)$ and~$w\tra x$ is~$O(n^2\phi)$. This upper bound holds regardless of whether the objective functions are to be maximized or minimized.
\end{theorem}

We will not prove Theorem~\ref{thm:UpperParetoGeneral}, but let us remark that its proof is very similar to the proof of Theorem~\ref{thm:UpperPareto}. In fact we never used in the proof that~$\calS=\{0,1\}^n$ and that~$p$ is linear. The fact that all weights~$w_i$ are positive was only used to argue that there must be an index~$i$ with~$x^{\star}_i=0$ and~$\hat{x}_i=1$. For general~$w_i$, it could also be the other way round. Handling this issue is the only modification of the proof that is not completely straightforward.

To illustrate the power of Theorem~\ref{thm:UpperParetoGeneral}, let us discuss its implications on graph problems. For a given graph with $m$ edges $e_1,\ldots,e_m$, one can identify every vector $x\in\{0,1\}^m$ with a subset of edges $E'=\{e_i\mid x_i=1\}$. Then $x$ is the so-called incidence vector of the edge set $E'$. If, for example, there is a source vertex~$s$ and a target vertex~$v$ given, one could choose the set~$\calS$ of feasible solutions as the set of all incidence vectors of paths from $s$ to $v$ in the given graph. This way, Theorem~\ref{thm:UpperParetoGeneral} implies that the expected number of Pareto-optimal $s$-$v$-paths in the bicriteria shortest-path problem is $O(m^2\phi)$. Similarly, one could choose $\calS$ as the set of incidence vectors of all spanning trees of a given graph. Then the result implies that in expectation there are only $O(m^2\phi)$ Pareto-optimal spanning trees in the bicriteria spanning tree problem. In the traveling salesman problem (TSP) we are given an undirected graph with edge weights and the goal is to find a shortest tour (i.e., Hamiltonian cycle) that visits every vertex exactly once. As for the bicritera shortest path problem, Theorem~\ref{thm:UpperParetoGeneral} implies that in expectation there are only $O(m^2\phi)$ Pareto-optimal tours in the bicriteria version of the TSP.

For the Bicriteria Bellman-Ford algorithm we obtain the following corollary.

\begin{corollary}
Consider an instance of the bicriteria shortest-path problem with arbitrary costs and non-negative $\phi$-perturbed weights from the interval~$[0,1]$. Let~$n$ and~$m$ denote the number of vertices and edges, respectively. Then the expected running time of the Bicriteria Bellman-Ford algorithm is~$O\left(nm^3\phi\right)$.
\end{corollary}
\begin{proof}
We can use Theorem~\ref{thm:UpperParetoGeneral} to bound the expected size of each list~$L_v^i$ that occurs throughout the algorithm by $O(m^2\phi)$, where~$m$ denotes the number of edges in the graph. Using linearity of expectation and Theorem~\ref{thm:BellmannFordRunningTime} yields that the expected running time is
\[
   \Theta\left(\sum_{i=1}^{M}\left(\Ex{|L_{u(R_i)}^{i-1}|}+\Ex{|L_{v(R_i)}^{i-1}|}\right)\right).
\]
Using that the expected length of each list is $O(m^2\phi)$ and $M=(n-1)\cdot m$ implies the claimed bound.
\end{proof}

Let us finally remark that Theorem~\ref{thm:UpperParetoGeneral} can also be adapted to the setting where the set~$\calS$ of feasible solutions is an arbitrary subset of $\{0,\ldots,k\}^n$ for some~$k\in\NN$. Then the expected number of Pareto-optimal solutions is $O(n^2k^2\phi)$. This is useful to model, for example, the bounded knapsack problem, in which a number of identical copies of every item is given.

\subsection{Multiobjective Optimization Problems}

Even though Theorem~\ref{thm:UpperParetoGeneral} is quite general, it still has one severe restriction: it only applies to optimization problems with two objective functions. The extension to optimization problems with more than two objectives is rather challenging and requires different methods. In this section, we summarize the main results.

In Theorem~\ref{thm:UpperParetoGeneral} one of the objective functions is assumed to be arbitrary while the other is linear with $\phi$-perturbed coefficients. We consider now optimization problems with one arbitrary objective function and~$d$ linear objective functions with~$\phi$-perturbed coefficients. \citet{RoeglinT09} were the first to study this model. They proved an upper bound of $O((n^2\phi)^{f(d)})$ for the expected number of Pareto-optimal solutions where~$f$ is a rapidly growing function (roughly~$2^d d!$). This has been improved by \citet{MoitraO12} to $O(n^{2d}\phi^{d(d+1)/2})$. \citet{BrunschR15} improved the upper bound further to $O(n^{2d}\phi^d)$ under the assumption that all density functions are unimodal, where a function~$f:\RR\to\RR$ is called unimodal if there exists some~$x\in\RR$ such that~$f$ is monotonically increasing on~$(-\infty,x]$ and monotonically decreasing on~$[x,\infty)$.

The $c$th moment of a random variable~$X$ is the expected value~$\Ex{X^c}$ if it exists. \citet{BrunschR15} also prove upper bounds on the moments of the number of Pareto-optimal solutions. In particular they show that for any constant~$c$ the $c$th moment is bounded from above by $O((n^{2d}\phi^{d(d+1)/2})^c)$ and $O((n^{2d}\phi^d)^c)$ for general and unimodal densities, respectively. Upper bounds for the moments give rise to non-trivial tail bounds. Consider the case~$d=1$. Then the $c$th moment is bounded from above by~$b_c(n^2\phi)^c$ for some constant~$b_c$ depending on~$c$. Applying Markov's inequality to the $c$th moment yields for every~$\alpha\ge 1$
\begin{align*}
  \Pr[|\calP| \ge \alpha\cdot (n^2\phi)]
  = \Pr[|\calP|^c \ge \alpha^c (n^2\phi)^c]
  = \Pr\bigg[|\calP|^c \ge \frac{\alpha^c}{b_c}\cdot b_c(n^2\phi)^c\bigg]
  \le \frac{b_c}{\alpha^c},
\end{align*}
while applying Markov's inequality directly to~$|\calP|$ yields only a bound of (roughly) $1/\alpha$. Upper bounds for the moments are also important for another reason: If the running time of an algorithm depends polynomially but not linearly on the number of Pareto-optimal solutions (like the running time of the Nemhauser-Ullmann algorithm for more than two objective functions), then Theorem~\ref{thm:UpperParetoGeneral} cannot be used to derive any bound on its expected running time. This is because a bound on~$\Ex{|\calP|}$ does not imply any bound on, for example, $\Ex{|\calP|^2}$. Only with the result of Brunsch and R\"oglin about the moments of~$|\calP|$ a polynomial bound follows for the expected running time of these algorithms.

Improving earlier work of \citet{BrunschGRR14}, \citet{BrunschPhD} shows lower bounds for the expected number of Pareto-optimal solutions of $\Omega(n^2\phi)$ for $d=1$ and $\Omega(n^{d-1.5}\phi^d)$ for~$d\ge 2$. Hence the upper bound in Theorem~\ref{thm:UpperParetoGeneral} for the bicriteria case is asymptotically tight. 

A $\phi$-perturbed number is non-zero with probability~1. This implies that each of the $d$ objective functions depends on all the variables. This limits the expressibility of the model because there are many examples of problems in which some objective function depends only on a certain subset of the variables. \citet{BrunschR15} discuss this subtle issue in more detail and they also give concrete examples. To circumvent this problem, they introduce \emph{zero-preserving perturbations}. In their model, the adversary can decide for each coefficient whether it should be a $\phi$-perturbed number or is deterministically set to zero. For this model they prove upper bounds of $O(n^{O(d^3)}\phi^d)$ and $O((n\phi)^{O(d^3)})$ for unimodal and general $\phi$-perturbed coefficients, respectively, for the expected number of Pareto-optimal solutions.

\section{Smoothed Complexity of Binary Optimization Problems}\label{sec:SmoothedComplexity}

The results on the expected number of Pareto-optimal solutions imply that the knapsack problem can be solved in expected polynomial time on instances with $\phi$-perturbed weights or $\phi$-perturbed profits (Corollary~\ref{cor:NUAlgo}). A natural question is whether or not similar results also hold for other NP-hard optimizations problems. Does, for example, the TSP admit an algorithm with expected polynomial running time if all distances are $\phi$-perturbed? Instead of studying each problem separately, we will now present a general characterization due to \citet{BeierV06} which combinatorial optimization problems can be solved efficiently on instances with $\phi$-perturbed numbers. 

While most smoothed analyses in the literature focus on the analysis of specific algorithms, this section instead considers \emph{problems} in the sense of complexity theory. We will study \emph{linear binary optimization problems}. In an instance of such a problem~$\Pi$, a linear objective function $c\tra x=c_1x_1+\cdots+c_nx_n$ is to be minimized or maximized over an arbitrary set~$\calS\subseteq\{0,1\}^n$ of feasible solutions. The problem~$\Pi$ could, for example, be the TSP and the coefficients~$c_i$ could be the edge lengths. (See also the discussion in Section~\ref{subsec:GeneralModel} on how graph problems can be encoded as binary optimization problems.) One could also encode the knapsack problem as a linear binary optimization problem. Then~$\calS$ contains all subsets of items whose total weight does not exceed the capacity.

We will study the \emph{smoothed complexity} of linear binary optimization problems, by which we mean the complexity of instances in which the coefficients~$c_1,\ldots,c_n$ are~$\phi$-perturbed numbers from the interval~$[-1,1]$. We will assume without loss of generality that the objective function~$c\tra x$ is to be minimized. Since $\phi$-perturbed numbers have infinite encoding length with probability~1, we have to discuss the machine model that we will use in the following. One could change the input model and assume that the $\phi$-perturbed coefficients are discretized by rounding them after a polynomial number, say $n^2$, of bits. The effect of this rounding is so small that it does not influence our results. We will, however, not make this assumption explicit and use, for the sake of simplicity, the continuous random variables in our probabilistic analysis. When defining the input size we will not take the encoding length of the coefficients~$c_i$ into account. Instead we assume that the coefficients~$c_1,\ldots,c_n$ contribute in total only~$n$ to the input length.

To state the main result, let us recall two definitions from computational complexity. We call a linear binary optimization problem \emph{strongly NP-hard} if it is already NP-hard when restricted to instances with integer coefficients~$c_i$ in which the largest absolute value~$C:=\max_i|c_i|$ of any of the coefficients is bounded by a polynomial in the input length. The TSP is, for example, strongly NP-hard because it is already NP-hard when all edges have length either~1 or~2. The knapsack problem, on the other hand, is not strongly NP-hard because instances in which all profits are integers and polynomially bounded in the input size can be solved by dynamic programming in polynomial time.

A language~$L$ belongs to the complexity class~ZPP (zero-error probabilistic polynomial time) if there exists a randomized algorithm~$A$ that decides for each input~$x$ in expected polynomial time whether or not~$x$ belongs to~$L$. That is,~$A$ always produces the correct answer but its running time is a random variable whose expected value is bounded polynomially for every input~$x$. Let us point out that the expectation is only with respect to the random decisions of the algorithm and not with respect to a randomly chosen input. It is yet unclear whether or not P$=$ZPP. In any case, languages that belong to ZPP are generally considered to be easy to decide and NP-hard problems are believed to not lie in ZPP.

\begin{theorem}\label{thm:SmoothedComplexityLower}
Let~$\Pi$ be a linear binary optimization problem that is strongly NP-hard. Then there does not exist an algorithm for~$\Pi$ whose expected running time is polynomially bounded in~$N$ and~$\phi$ for instances with $\phi$-perturbed coefficients from~$[-1,1]$, where~$N$ denotes the input length, unless~$\text{NP}\subseteq\text{ZPP}$.
\end{theorem}

The main idea of the proof of this theorem can be summarized as follows: An algorithm~$A$ for~$\Pi$ with expected running time polynomial in~$N$ and~$\phi$ can be used to solve worst-case instances of~$\Pi$ with polynomially bounded numbers optimally in expected polynomial time. Given such a worst-case instance, one could add a small amount of random noise to all the numbers and then solve the resulting instance with~$A$ in expected time polynomial in~$N$ and~$\phi$. If this random noise is small enough ($\phi=\Theta(C)$) then it does not change the optimal solution. This way, we obtain an algorithm that solves worst-case instances with polynomially bounded numbers in expected polynomial time, implying that~$\text{NP}\subseteq\text{ZPP}$.

The previous theorem shows that $\phi$-perturbed instances of strongly NP-hard optimization problems are not easier to solve than worst-case instances. Hence, these problems stay hard also in the model of smoothed analysis. One consequence of this result is that there is no hope that the TSP can be solved efficiently when the edge lengths are randomly perturbed. This is in clear contrast to the knapsack problem, which is easy to solve on randomly perturbed inputs. We will now state a more general positive result. We say that a linear binary optimization problem~$\Pi$ can be solved in \emph{pseudo-linear time} if there exists an algorithm whose running time on instances with integer coefficients is bounded from above by~$p(N)\cdot C$, where~$p$ denotes a polynomial, $N$ denotes the input length, and~$C$ denotes the largest absolute value of any of the coefficients.

\begin{theorem}\label{thm:SmoothedComplexityUpper}
A linear binary optimization problem~$\Pi$ that can be solved in pseudo-linear time in the worst case can be solved in expected polynomial time (with respect to the input length and~$\phi)$ on instances with $\phi$-perturbed numbers from~$[-1,1]$.
\end{theorem}

Let~$A_{\text{p}}$ be an algorithm that solves integral instances of~$\Pi$ in pseudo-linear time. In the proof of Theorem~\ref{thm:SmoothedComplexityUpper}, the algorithm~$A_{\text{p}}$ is used to construct an algorithm~$A$ that solves instances with $\phi$-perturbed numbers in expected polynomial time. Algorithm~$A$ first rounds all $\phi$-perturbed coefficients after some number~$b$ of bits after the binary point. Then it uses the algorithm~$A_{\text{p}}$ to solve the rounded instance. One can prove that for~$b=\Theta(\log{n})$ rounding all coefficients does not change the optimal solution with high probability. This is based on the observation that in instances with $\phi$-perturbed numbers usually the best solution is significantly better than the second best solution and hence it stays optimal even after rounding all coefficients (see Exercise~\ref{exer:isolation2}). For~$b=\Theta(\log{n})$ the running time of~$A_{\text{p}}$ to solve the rounded instance optimally is polynomial. This yields an algorithm that always runs
in polynomial time and solves $\phi$-perturbed instances of~$\Pi$ with high probability correctly.
It is possible to adapt this approach to obtain an algorithm that always computes the optimal solution and whose expected running time is polynomial.

\section{Conclusions}

We have proven bounds on the expected number of Pareto-optimal solutions and we have studied the complexity of linear binary optimization problems in the framework of smoothed analysis. Our results are in many cases consistent with empirical observations. The knapsack problem is, for example, easy to solve in applications and has few Pareto-optimal solutions while solving large-scale TSP instances optimally is computationally still expensive despite a lot of progress that has been made in the last decades and great speedups in the common solvers.

The models that we considered in this chapter are very general, in particular because the set $\calS$ of feasible solutions can be arbitrarily chosen, both in Section~\ref{sec:NumberPOSolutions} and in Section~\ref{sec:SmoothedComplexity}. However, this generality is also a drawback of our results because the adversary is still rather powerful and can exactly determine the combinatorial structure of the problem. Often problems are easier in applications than in the worst case because the instances obey certain structural properties. Depending on the problem and application, input graphs might be planar or have small degree, distances might satisfy the triangle inequality etc. Such structural properties are not considered in our general model. Hence, often it is advisable to look in more detail into the instances that are really relevant in applications instead of only assuming that some coefficients are random.

An illustrative experimental study of the multiobjective shortest path problem is due to \citet{Mueller-Hannemann2006}. They consider a graph that is obtained from the daily train schedule of the German railway network and observe that the number of Pareto-optimal train connections in view of travel time, fare, and number of train changes is very small (for no pair of nodes there were more than~$8$ Pareto-optimal connections in the experiments). This is much smaller than suggested by Theorem~\ref{thm:UpperParetoGeneral}. One possible explanation is that in this and many other applications, the objective functions are not independent but to some degree correlated, which might reduce the number of Pareto-optimal solutions. It would be interesting to find a formal model for correlated objective functions that explains the extremely small number of Pareto-optimal solutions observed in this setting.

\section*{Notes}

The Bicriteria Bellman-Ford algorithm was described by \citet{CorleyM85}. The analysis of its running time presented in this chapter can also be found in \citet{BeierPhD}. \citet{BeierV04} initiated the study of the number of Pareto-optimal solutions in the framework of smoothed analysis. The proof of Theorem~\ref{thm:UpperPareto} in this chapter follows an improved and simplified analysis due to \citet{BeierRV07}. This analysis also generalizes the original work of Beier and V\"ocking to integer optimization problems. The bound stated in \citet{BeierRV07} is $O(n^2k^2\log(k)\phi)$ if $\calS\subseteq\{0,\ldots,k\}^n$. It has been improved to $O(n^2k^2\phi)$ by \citet{RoglinR17}.

The results in Section~\ref{sec:SmoothedComplexity} can be found in~\citet{BeierV06}. Theorems~\ref{thm:SmoothedComplexityLower} and~\ref{thm:SmoothedComplexityUpper} do not give a complete characterization of the smoothed complexity of linear binary optimization problems because Theorem~\ref{thm:SmoothedComplexityUpper} does only apply to pseudo-linear and not to general pseudo-polynomial algorithms. Beier and V\"ocking circumvent this problem by introducing a notion of polynomial smoothed complexity that is not based on expected running times (similar to polynomial average-case complexity). Later \citet{RoeglinT09} showed that all problems that can be solved in pseudo-polynomial time in the worst case can be solved in expected polynomial time on $\phi$-perturbed instances, which completes the characterization.

\section*{Exercises}
\begin{enumerate}

\item \label{exer:NUAlgo}
Implement the Nemhauser-Ullmann algorithm so that your implementation achieves a running time of $\Theta(\sum_{i=0}^{n-1}|\mathcal{P}_i|)$.

\item Find an instance of the knapsack problem with~$|\calP_{i+1}| < |\calP_{i}|$ for some~$i$.

\item \label{exer:NumberPareto} Construct instances for the bicriteria shortest path problem with an exponential number of Pareto-optimal $s$-$v$-paths for some vertices~$s$ and~$v$.

\item \label{exer:BellmanFord} Prove that the Bicriteria Bellman-Ford algorithm is correct, i.e., that after termination the list~$L_v$ equals for every vertex~$v\in V$ the set of Pareto-optimal $s$-$v$-paths

\item \label{exer:FloydWarshall} A famous algorithm for the single criterion all-pairs shortest path problem is the Floyd-Warshall algorithm. Adapt this algorithm to the bicriteria all-pairs shortest path problem (given a graph~$G$ with costs and weights, compute for each pair $(u,v)$ of vertices the set of Pareto-optimal $u$-$v$-paths in~$G$). State a bound on its running time in the same fashion as Theorem~\ref{thm:BellmannFordRunningTime}. What is the expected running time if the weights are $\phi$-perturbed?

\item The concept of zero-preserving perturbations could also be applied to the bicriteria case with one adversarial objective function and one linear objective function with $\phi$-perturbed coefficients. Show that, in contrast to the multiobjective case, for bicriteria optimization problems it does not increase the expressibility. For this, show that zero-preserving perturbations for the bicriteria case can be simulated by $\phi$-perturbed coefficients if the set~$\calS$ of feasible solutions is adapted appropriately. Why does this simulation not work for problems with three or more objectives?

\item Prove that the expected number of Pareto-optimal points among $n$ points drawn independently and uniformly at random from the unit square is~$O(\log{n})$.

\item \label{exer:isolation2} Given an instance~$\calI$ of some linear binary optimization problem~$\Pi$ with a set~$\calS\subseteq\{0,1\}^n$ of feasible solutions,
the winner gap is defined as
\[
   \Delta = c\tra x^{\star\star} - c\tra x^\star,
\]
where
\[
   x^\star = \argmin\{c\tra x\mid x\in\calS\}\quad\text{and}\quad
   x^{\star\star} = \argmin\{c\tra x\mid x\in\calS\setminus\{x^\star\}\}
\]
denote the best and second best solution of~$\calI$, respectively.
Let~$\calI$ be an instance of~$\Pi$ with $\phi$-perturbed coefficients~$c_1,\ldots,c_n$.
Prove that, for every~$\eps>0$,
\[
   \Pr[\Delta \le \eps] \le 2n\phi\eps.
\]
\emph{Hint:} This statement follows by similar arguments as Lemma~\ref{lemma:ProbPOInterval}.

\end{enumerate}

\bibliography{Chapter15}

\begin{thebibliography}{20}
\expandafter\ifx\csname natexlab\endcsname\relax\def\natexlab#1{#1}\fi
\expandafter\ifx\csname selectlanguage\endcsname\relax
  \def\selectlanguage#1{\relax}\fi

\bibitem[\protect\citename{Beier, }2004]{BeierPhD}
Beier, Ren{\'e}. 2004.
\newblock {\em Probabilistic Analysis of Discrete Optimization Problems}.
\newblock Ph.D. thesis, Universit{\"a}t des Saarlandes.

\bibitem[\protect\citename{Beier and V{\"o}cking, }2004]{BeierV04}
Beier, Ren{\'e}, and V{\"o}cking, Berthold. 2004.
\newblock Random Knapsack in Expected Polynomial Time.
\newblock {\em Journal of Computer and System Sciences}, {\bf 69}(3), 306--329.

\bibitem[\protect\citename{Beier and V{\"o}cking, }2006]{BeierV06}
Beier, Ren{\'e}, and V{\"o}cking, Berthold. 2006.
\newblock Typical Properties of Winners and Losers in Discrete Optimization.
\newblock {\em SIAM Journal on Computing}, {\bf 35}(4), 855--881.

\bibitem[\protect\citename{Beier {et~al.}, }2007]{BeierRV07}
Beier, Ren{\'e}, R{\"o}glin, Heiko, and V{\"o}cking, Berthold. 2007.
\newblock The Smoothed Number of {P}areto Optimal Solutions in Bicriteria
  Integer Optimization.
\newblock {Pages  53--67 of:} {\em Proceedings of the 12th International
  Conference on Integer Programming and Combinatorial Optimization (IPCO)}.

\bibitem[\protect\citename{Bökler {et~al.}, }2017]{BoeklerEMM17}
Bökler, Fritz, Ehrgott, Matthias, Morris, Christopher, and Mutzel, Petra.
  2017.
\newblock Output-sensitive complexity of multiobjective combinatorial
  optimization.
\newblock {\em Journal of Multi-Criteria Decision Analysis}, {\bf 24}(1-2),
  25--36.

\bibitem[\protect\citename{Brunsch, }2014]{BrunschPhD}
Brunsch, Tobias. 2014.
\newblock {\em Smoothed Analysis of Selected Optimization Problems and
  Algorithms}.
\newblock Ph.D. thesis, Universit{\"a}t Bonn.

\bibitem[\protect\citename{Brunsch and R{\"{o}}glin, }2015]{BrunschR15}
Brunsch, Tobias, and R{\"{o}}glin, Heiko. 2015.
\newblock Improved Smoothed Analysis of Multiobjective Optimization.
\newblock {\em Journal of the ACM}, {\bf 62}(1), 4:1--4:58.

\bibitem[\protect\citename{Brunsch {et~al.}, }2014]{BrunschGRR14}
Brunsch, Tobias, Goyal, Navin, Rademacher, Luis, and R{\"{o}}glin, Heiko. 2014.
\newblock Lower Bounds for the Average and Smoothed Number of Pareto-Optima.
\newblock {\em Theory of Computing}, {\bf 10}, 237--256.

\bibitem[\protect\citename{Corley and Moon, }1985]{CorleyM85}
Corley, H.~William, and Moon, I.~Douglas. 1985.
\newblock Shortest paths in networks with vector weights.
\newblock {\em Journal of Optimization Theory and Application}, {\bf 46}(1),
  79--86.

\bibitem[\protect\citename{Diakonikolas and Yannakakis, }2007]{DiakonikolasY07}
Diakonikolas, Ilias, and Yannakakis, Mihalis. 2007.
\newblock Small Approximate Pareto Sets for Bi-objective Shortest Paths and
  Other Problems.
\newblock {Pages  74--88 of:} {\em Proceedings of the 10th International
  Workshop on Approximation Algorithms for Combinatorial Optimization Problems
  (APPROX)}.

\bibitem[\protect\citename{Hansen, }1979]{Hansen79}
Hansen, Pierre. 1979.
\newblock Bicriterion Path Problems.
\newblock {Pages  109--127 of:} {\em Multiple Criteria Decision Making: Theory
  and Applications}.
\newblock Lecture Notes in Economics and Mathematical Systems, vol. 177.

\bibitem[\protect\citename{Hansen, }1980]{Hansen80}
Hansen, Pierre. 1980.
\newblock Bicriterion Path Problems.
\newblock {Pages  109--127 of:} {\em Multiple Criteria Decision Making: Theory
  and Applications}.
\newblock Lecture Notes in Economics and Mathematical Systems, vol. 177.

\bibitem[\protect\citename{Kung {et~al.}, }1975]{KungLP75}
Kung, H.~T., Luccio, Fabrizio, and Preparata, Franco~P. 1975.
\newblock On Finding the Maxima of a Set of Vectors.
\newblock {\em Journal of the ACM}, {\bf 22}(4), 469--476.

\bibitem[\protect\citename{Moitra and O'Donnell, }2012]{MoitraO12}
Moitra, Ankur, and O'Donnell, Ryan. 2012.
\newblock Pareto Optimal Solutions for Smoothed Analysts.
\newblock {\em SIAM Journal on Computing}, {\bf 41}(5), 1266--1284.

\bibitem[\protect\citename{M{\"u}ller-Hannemann and Weihe,
  }2006]{Mueller-Hannemann2006}
M{\"u}ller-Hannemann, Matthias, and Weihe, Karsten. 2006.
\newblock On the cardinality of the Pareto set in bicriteria shortest path
  problems.
\newblock {\em Annals of Operations Research}, {\bf 147}(1), 269--286.

\bibitem[\protect\citename{Nemhauser and Ullmann, }1969]{NU69}
Nemhauser, George~L., and Ullmann, Zev. 1969.
\newblock Discrete dynamic programming and capital allocation.
\newblock {\em Management Science}, {\bf 15}(9), 494--505.

\bibitem[\protect\citename{Papadimitriou and Yannakakis,
  }2000]{PapadimitriouY00}
Papadimitriou, Christos~H., and Yannakakis, Mihalis. 2000.
\newblock On the approximability of trade-offs and optimal access of Web
  sources.
\newblock {Pages  86--92 of:} {\em Proceedings of the 41st Annual IEEE
  Symposium on Foundations of Computer Science (FOCS)}.

\bibitem[\protect\citename{R{\"{o}}glin and R{\"{o}}sner, }2017]{RoglinR17}
R{\"{o}}glin, Heiko, and R{\"{o}}sner, Clemens. 2017.
\newblock The Smoothed Number of Pareto-Optimal Solutions in Non-integer
  Bicriteria Optimization.
\newblock {Pages  543--555 of:} {\em Proceedings of the 14th Annual Conference
  on Theory and Applications of Models of Computation (TAMC)}.

\bibitem[\protect\citename{R{\"o}glin and Teng, }2009]{RoeglinT09}
R{\"o}glin, Heiko, and Teng, Shang-Hua. 2009.
\newblock Smoothed Analysis of Multiobjective Optimization.
\newblock {Pages  681--690 of:} {\em Proceedings of the 50th Annual IEEE
  Symposium on Foundations of Computer Science (FOCS)}.

\bibitem[\protect\citename{Vassilvitskii and Yannakakis,
  }2005]{VassilvitskiiY04}
Vassilvitskii, Sergei, and Yannakakis, Mihalis. 2005.
\newblock Efficiently computing succinct trade-off curves.
\newblock {\em Theoretical Computer Science}, {\bf 348}(2--3), 334--356.

\end{thebibliography}
\bibliographystyle{cambridgeauthordate}

\end{document}